\documentclass[]{article}

\usepackage{spur}
\usepackage[headheight = 14pt,margin=1in]{geometry}
\usepackage[shortlabels]{enumitem}

\title{Simple Eigenvalues and Non-vanishing Eigenvectors of the Anderson Model}
\author{Oluyinka Lindblad and Ezra Guerrero}
\date{}

\theoremstyle{plain}
\newtheorem{theorem}{Theorem}[section]

\theoremstyle{plain}
\newtheorem{lemma}[theorem]{Lemma}

\theoremstyle{plain}
\newtheorem{proposition}[theorem]{Proposition}

\theoremstyle{plain}
\newtheorem{corollary}[theorem]{Corollary}

\theoremstyle{definition}
\newtheorem{definition}[theorem]{Definition}

\theoremstyle{plain}

\begin{document}

\maketitle

\begin{abstract}
    We consider the Anderson model on the finite grid $G = \ZZ/L_1\ZZ\times\cdots\times\ZZ/L_d\ZZ$, defined by the random Hamiltonian $H_t=\Delta+tV$, where $\Delta$ is the discrete Laplacian and $V=\mathrm{diag}(\{\omega_{x}\}_{x\in G})$ is a random onsite potential with $\omega_x\sim\mu$ i.i.d. We ask the natural question of when $H_t$ has simple eigenvalues and non-vanishing eigenvectors. We prove that, when $\mu$ is a continuous probability distribution, $H_t$ has this property for all but finitely many $t$ values with probability $1$. However, when $\mu$ is a Bernoulli distribution, the conditions fail with positive probability, for which we give a lower bound. We also calculate the exact probability of these conditions being met in the Bernoulli case when $d = 1$ and $L = L_1$ is prime.
\end{abstract}

\section{Introduction}
\label{section_introduction}

The Anderson model, introduced in 1958 by P.W. Anderson \cite{anderson1958absence} is the toy model for the study of localization of electrons in disordered media. If the strength of the disorder is large enough, eigenvectors in such systems become localized \cite{Abou-Chacra_1973,aizenmanWarzel2009localization}. Therefore, we expect to see a transition from delocalized states to localized states as the strength of the disorder increases. This transition leads to qualitative changes in the physics of the system. For example, for an electron in a conducting sheet, this transition implies a conductor-insulator phase transition. 

Concretely, the Anderson model is given by a random Hamiltonian on the integer lattice $\ZZ^d$ defined as
\[ H_t := \Delta + t V, \]
where $\Delta$ is the discrete Laplacian on $\ZZ^d$, $t\in\RR$ is a coupling constant, and $V= \diag(\{\omega_x\}_{x\in\ZZ^d})$ is an onsite potential given by independent identically distributed (i.i.d.) random variables $\omega_x$ with probability distribution $\mu$. 

\begin{figure}[h!]
    \centering
    \begin{minipage}{0.45\linewidth}
        \centering
        \includegraphics[scale=0.55]{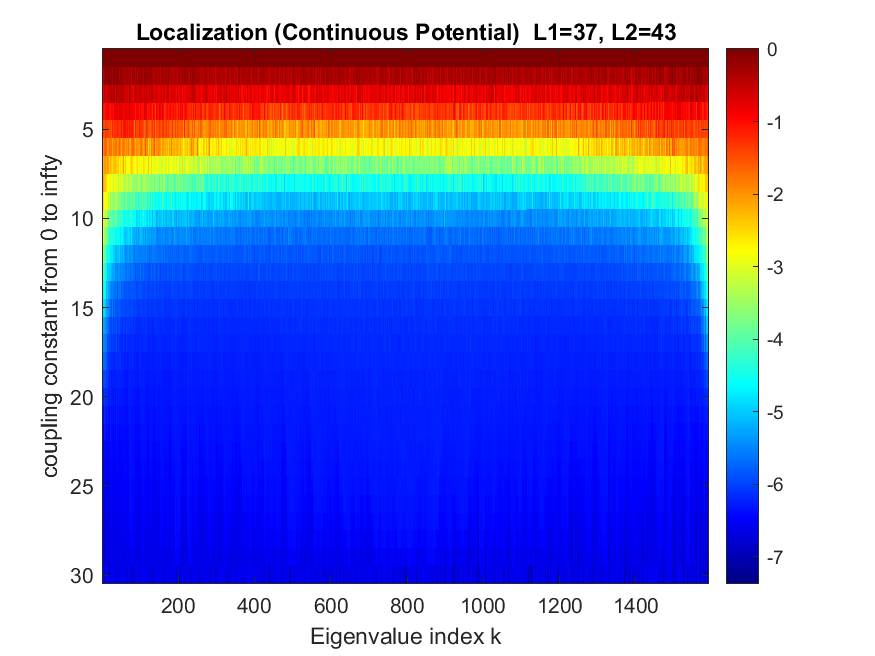}
    \end{minipage}
    \begin{minipage}{0.45\linewidth}
        \centering
        \includegraphics[scale=0.4]{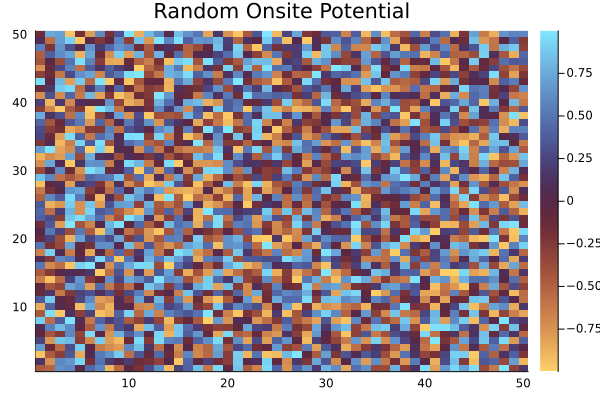}
    \end{minipage}
    \caption{The localization-delocalization phase transition: $\log(\|\varphi_k\|_{4}^4)$ for the $k$-th eigenvector of $H_t$, as a function of $k$ on the x-axis and $t$ on the y-axis (left) and a random onsite potential (right), with $\mu\sim\operatorname{Uniform}[-1,1]$. }
\end{figure}

In order to tractably study the localization properties of the Anderson model, we consider restrictions to a finite grid \cite{klein2007multiscaleanalysislocalizationrandom}. In this paper, we consider periodic boundary conditions, so our graph is
\[ G := (\ZZ/L_1\ZZ)\times(\ZZ/L_2\ZZ)\times\cdots\times(\ZZ/L_d\ZZ), \]
where $L_1,L_2,\ldots,L_d>2$ are the side lengths of the grid. Then, the random Hamiltonian is given by
\[ (H_t\psi)(j) = \sum_{k\sim j}(\psi(j)-\psi(k)) + t\omega_j\psi(j), \]
where $k\sim j$ means $k$ and $j$ are adjacent in $G$. We are interested in the spectral behavior of $H_t$ in the limit as $L_1,L_2,\ldots,L_d$ approach infinity uniformly.

In this paper, we focus on two natural spectral conditions on the random Hamiltonian. We study the probability that $H_t$ both (a) has simple eigenvalues and (b) has all entries of all eigenvectors of $H_t$ nonzero. In the case of high disorder, Theorem 1.6 in \cite{ELGART20163465} shows that (a) holds with high probability. Moreover, in this regime, we expect localized eigenvectors to decay exponentially. A common question is the behavior of this exponential decay \cite{Ding_Smart_2019}. For example, the Landis conjecture \cite{logunov2020landisconjectureexponentialdecay} states that sufficiently fast exponential decay leads to compactly supported eigenvectors. Condition (b) asks a more general question --- when do the eigenvectors have any zeros at all?
Moreover, when conditions (a) and (b) hold, we can define the nodal count as in \cite{AlonGoresky}. Experimental data suggests that the variance of this quantity is a predictor for eigenvector localization. We expect that understanding these conditions will be fruitful in the study of Anderson localization.

If a matrix satisfies (b) we say that it has \textit{non-vanishing eigenvectors}. Otherwise, we say it has some \textit{vanishing eigenvector}. We present succinct proofs on the genericity of our conditions when $\mu$ is continuous and on the probability of our conditions being met when $\mu$ is discrete. 
We state our main results below:

\subsection{Main Results}
\label{subsection_main-results}

\begin{definition}[Good/Bad Potentials]
\label{definition_good-bad-potentials}
    We say that $V$ is a good potential if $H_t=\Delta+tV$ has simple eigenvalues and non-vanishing eigenvectors for all but finitely many $t$ values. We say that $V$ is a bad potential if $H_t$ fails to satisfy at least one of these conditions for any $t\in\RR$.
\end{definition}

It is not clear that every potential must be either good or bad. In light of Proposition \ref{proposition_all-but-finitely}, we see that every potential must fall into exactly one of these classifications.

In section \ref{section_continuous}, we establish results for general perturbations and use them to prove the following theorem:

\begin{restatable}{theorem}{theoremcontinuousprob}
\label{theorem_continuous-prob}
    Let $G$ be any connected graph. If $\mu$ is absolutely continuous with respect to the Lebesgue measure, then $\PP(V\text{ is good})=1$. In particular, for the Anderson model, $\PP(V\text{ is good})=1$ for any $d$ and any side lengths $L_1,\ldots,L_d$.
\end{restatable}

The key insight of the proof is that $V$ being bad leads to an algebraic condition on its entries. In fact, the result holds for any $\mu$ which, with probability $1$, assigns the potential algebraically independent diagonal entries. 

This is not the case when $\mu$ is a discrete probability distribution. 
To explain the discrepancy, in section \ref{section_symmetry} we show \textit{shared symmetries} --- orthogonal matrices that commute with both $\Delta$ and $V$ --- lead to the failure of our conditions.
In section \ref{section_discrete}, we return to the Anderson model and use our results to establish lower bounds on the probability that $V$ is bad. 
In what follows, suppose that $\mu$ is the distribution
\[ \mu = \begin{cases}
    1,&\text{with probability }p \\
    -1, &\text{with probability }1-p
\end{cases}. \]
We prove that
\begin{restatable}{theorem}{theoremlowerboundgeneral}
\label{theorem_lower-bound-general}
    For any $d$ and any $L_1,\ldots,L_d$, we have
    \[ \PP(V\text{ is bad}) \ge \sum_{k=1}^d\left[(-1)^{k-1}\cdot\sum_{j_1<\ldots<j_k}\left(p^{L_{j_1}\cdots L_{j_k}} + (1-p)^{L_{j_1}\cdots L_{j_k}}\right)^{\frac n{L_{j_1}\cdots L_{j_k}}}\right] > 0. \]
\end{restatable}

We prove this by estimating the number of potentials that have certain shared symmetries with $\Delta$. Conversely, we also ask whether every bad potential shares a nontrivial symmetry with the laplacian. 

For a one-dimensional grid, when $L=L_1$ is prime, we prove that in fact, a potential is bad if and only if it has a reflection symmetry about some vertex. From this, we are able to count the exact number of bad potentials and study the behavior in the limit as $L \rightarrow \infty$.


\begin{restatable}{theorem}{theoremexactprobability}
\label{theorem_exact-probability}
    If $d=1$ and the side length, $L$, is prime, then
    \[ \PP(V\text{ is bad)} = L(p^2+(1-p)^2)^{\frac{L-1}2} - (L-1)(p^L+(1-p)^L). \]
    In particular, for $p\ne0,1$, $\PP(V\text{ is good})\to 1$ as we take $L\to\oo$.
\end{restatable}

\subsection{Conventions}

Unless otherwise stated, we work over the real vector space $\RR^n$. Let $\{e_k\}_{k=1}^n$ denote the standard basis and $\langle\cdot,\cdot\rangle$ denote the standard inner product. 
We write $I$ for the $n\times n$ identity matrix. We denote the entries of a vector $\varphi$ in the standard basis by $\varphi(i) = \langle \varphi, e_i \rangle$. Likewise, if $A$ is a matrix, we denote the entries of $A$ by $A(i,j) = \langle e_i, A e_j \rangle$.

\begin{definition}
\label{definition_support-matrix}
    Let $A$ be a real symmetric $n\times n$ matrix. We define the \textit{support} of $A$, $\supp(A)$, as the simple graph on $n$ vertices satisfying that there is an edge $i\sim j$ if and only if $A(i,j)=A(j,i)\ne 0$ for $i\ne j$.
\end{definition}

For example, if $G$ is a graph and $\Delta$ is its Laplacian matrix, $\supp(\Delta) = G$.

\begin{definition}
\label{definition_support-vector}
    Let $G$ be a graph with $n$ vertices and $\varphi$ a vector in $\RR^n$. We define the \textit{support} of $\varphi$ as $\supp(\varphi) = \{ i \in G | \varphi(i) \not= 0 \}$.
\end{definition}

\begin{definition}
\label{definition_graph-distance}
    Let $G$ be a connected graph with $n$ vertices and let $A_G$ be its adjacency matrix. Let $i$ and $j$ be vertices of $G$. Then, we denote by $d(i,j)$ the distance from $i$ to $j$. That is, 
    \[ d(i,j) := \min\{k\in\NN\mid A_G^k(i,j)\ne0\}. \]
    Moreover, for any vertex $i$ and $r\ge 0$, we denote by $B_r(i)$ the ball of radius $r$ around $i$. That is, 
    \[ B_r(i) := \{j\in G\mid d(i,j)\le r\}. \]
\end{definition}

\section{Continuous Probability Distributions}
\label{section_continuous}

\subsection{General Perturbations}
\label{subsection_general-perturbations}

In order to study matrices of the form $\Delta + t V$ where $V$ is diagonal, it will prove useful to analyze more general classes of matrices $A+tB$. In all that follows, define the one-parameter family
\[ H_t(A,B) := A+tB. \]
When the matrix $A$ is clear from context, we abbreviate to $H_t(B)$. If $A$ and $B$ are both clear, we abbreviate to $H_t$.

Our work is based on the following lemma for analytic one-parameter families of matrices, due to Kato (Theorem 1.8 in Chapter II of \cite{kato2013perturbation}) and Rellich (Theorem 1 in \cite{rellich1969perturbation}): 

\begin{lemma}[Kato-Rellich]
\label{lemma_kato-rellich}
    A one-parameter family of normal matrices $A_t$ that depends analytically on $t$ admits an orthonormal basis of eigenvectors that depends analytically on $t$. Furthermore, their corresponding eigenvalues also depend analytically on $t$.
\end{lemma}

We are interested in the spectral behavior of $H_t(A,B)$. We separate pairs of matrices into ``good'' pairs and ``bad'' pairs. 

\begin{definition}[Good/Bad Pairs]
\label{definition_good-bad-pairs}
    Let $A$ and $B$ be real symmetric $n\times n$ matrices. We say that $(A,B)$ is a good pair if $H_t(A,B)$ has simple eigenvalues and non-vanishing eigenvectors for all but finitely many $t$. We say that $(A,B)$ is a bad pair if $H_t$ fails to satisfy at least one of these conditions for any $t\in\RR$.
\end{definition}

As mentioned in the introduction, it is not obvious that every pair of matrices is either good or bad. The following proposition exploits the analiticity given by Lemma \ref{lemma_kato-rellich} to show good and bad pairs is all there is:

\begin{proposition}
\label{proposition_all-but-finitely}
    Let $A$ and $B$ be real symmetric $n\times n$ matrices. Let $\{\varphi_{k,t}\}_{k=1}^n$ be the analytic orthonormal basis of eigenvectors of $H_t(A,B)$ given by Lemma \ref{lemma_kato-rellich} and $\{\lambda_{k,t}\}_{k=1}^n$ the corresponding analytic eigenvalues. Take any $i,k\ne k'\in\{1,\ldots,n\}$. If either of $\lambda_{k,t}-\lambda_{k',t}\ne0$ or $\varphi_{k,t}(i)\ne0$ holds at some $t$, then it holds for all but finitely many $t\in\RR$.
\end{proposition}
\begin{proof}
    Consider the auxiliary one-parameter perturbation
    \[ A_s = (1-s)A + sB \]
    for $s\in[0,1]$. Observe that
    \[ A_s = (1-s)\left(A + \frac{s}{1-s}B\right) = (1-s)H_{s/(1-s)}. \]
    Therefore, $\{\varphi_{k,s/(1-s)}\}_{k=1}^n$ is an orthonormal eigenbasis of $A_s$ with corresponding eigenvalues $\{(1-s)\lambda_{k,s/(1-s)}\}_{k=1}^n$. Now, notice that the functions $(1-s)(\lambda_{k,s/(1-s)}-\lambda_{k',s/(1-s)})$ and $\varphi_{k,s/(1-s)}(i)$ are analytic functions of $s$ on the interval $[0,1]$. By compactness, these functions either have finitely many zeros or vanish for all $s\in[0,1]$. Hence, if either of the functions $\lambda_{k,t}-\lambda_{k',t}$ or $\varphi_{k,t}(i)$ is not identically zero, it has finitely many zeros for $t\in[0,\oo)$. An analogous computation gives finitely many zeros for $t\in(-\oo,0]$. Taking the union, we get finitely many zeros in $\RR$.
\end{proof}

\textbf{Remark}: From the auxiliary perturbation in the proof of Proposition \ref{proposition_all-but-finitely}, we learn that our conditions are invariant under nonzero scalings. In particular, $(A,B)$ is a good (resp. bad) pair if and only if $(B,A)$ is a good (resp. bad) pair.

Note that a real symmetric matrix $A$ has non-vanishing eigenvectors if and only if $\supp(\varphi) = \supp(A)$ for all eigenvectors $\varphi$. The following lemma characterizes the behavior of $\supp(\varphi)$ for some special perturbations:

\begin{lemma}
\label{lemma_support}
    Let $A$ be a real symmetric $n\times n$ matrix with connected support and let $D$ be a diagonal matrix with distinct diagonal real entries $x_1,\ldots,x_n$. Let 
    \[ \varphi_{k,t} = e_k + \sum_{j=1}^\oo t^j\varphi_k^{(j)} \]
    be the analytic orthonormal basis of eigenvectors of $H_t(D,A)$ given by Lemma \ref{lemma_kato-rellich}. Then, for all $j$, we have 
    \[ \supp\left(\varphi_k^{(j)}\right) \subseteq B_j(k). \]
    Furthermore, for $i\notin B_{j-1}(k)$,
    \begin{align}
        \label{equation_perturbation-entry}
        \varphi_k^{(j)}(i) = (C^j)(i,k),
    \end{align}
    where $C = -(D-x_kI)^+A$ and $(D-x_kI)^+$ is the diagonal matrix with $0$ at the $k-$th entry and $(x_m-x_k)^{-1}$ at the $m-$th entry, for $m\ne k$.
\end{lemma}
\begin{proof}
    Denote by
    \[ \lambda_{k,t} = x_k + \sum_{j=1}^\oo t^j\lambda_k^{(j)} \]
    the analytic eigenvalues corresponding to $\varphi_{k,t}$. Consider the eigenvalue equation
    \[ (D+tA)\varphi_{k,t} = \lambda_{k,t}\varphi_{k,t}.\]
    Substituting our power series expansions and equating terms with the same power of $t$, we find that there exist real constants $c^{(j)}$ and vectors 
    \[ v^{(j)}=\sum_{m=1}^j\lambda_k^{(m)}(D-x_kI)^+\varphi_k^{(j-m)} + c^{(j)}e_k\]
    such that\footnote{See, for example, Lemma 2.3 in \cite{alon2024average}}
    \[ \varphi_k^{(j)} = C\varphi_k^{(j-1)} + v^{(j)}. \]
    We proceed by induction on $j$. Note that $\varphi_k^{(0)} = e_k$, so $\supp\left(\varphi_k^{(0)}\right)=\{k\}=B_0(k)$. Additionally, $(C^0)(i,k)=\delta_{i,k}=e_k(i)$. This establishes our base case.

    Now, assume the result holds for all $j'<j$. Since $(D-x_kI)^+$ is diagonal, we have
    \begin{align}
    \label{equation_supp-vj}
    \supp\left(v^{(j)}\right) \subseteq \bigcup_{r=0}^{j-1}\supp\left(\varphi_k^{(r)}\right)\subseteq B_{j-1}(k),
    \end{align}
    where the last inclusion is by the induction hypothesis. Next,
    \[ \left(A\varphi_k^{(j-1)}\right)(i) = \sum_{r\sim i}A(i,r)\varphi_k^{(j-1)}(r), \]
    where $r\sim i$ means $r$ and $i$ are adjacent in $\supp(A)$. We deduce that the $i-$th entry of $A\varphi_k^{(j-1)}$ can only be nonzero if $i$ is adjacent to some vertex in $\supp\left(\varphi_k^{(j-1)}\right)\subseteq B_{j-1}(k)$. Thus, $\supp\left(A\varphi_k^{(j-1)}\right)\subseteq B_j(k)$. We obtain
    \[ \supp\left(C\varphi_k^{(j-1)}\right)\subseteq\supp\left(A\varphi_k^{(j-1)}\right)\subseteq B_j(k). \]
    Putting this and equation (\ref{equation_supp-vj}) together, we conclude that $\supp\left(\varphi_k^{(j)}\right)\subseteq B_j(k)$.

    Suppose that $i\notin B_{j-1}(k)$. From above, we have that $v^{j}(i)=0$. Moreover, by definition of $C$, if $C(i,r)\ne0$, then $A(i,r)\ne0$. Thus, $d(i,r)\le 1$. By the triangle inequality, we conclude $r\notin B_{j-2}(k)$.
    Using the inductive hypothesis,
    \begin{align*}
        \varphi_k^{(j)}(i) &= \sum_{r\sim i}C(i,r)\varphi_k^{(j-1)}(r) = \sum_{r\sim i}C(i,r)(C^{j-1})(r,k) = (C^j)(i,k).
    \end{align*}
    This concludes the induction.
\end{proof}

We are now ready to prove a useful proposition that gives us an infinite family of good pairs for any finite connected graph. Recall that a finite set $\{x_1,\ldots,x_n\}$ of real numbers is \textit{algebraically independent} if the only polynomial $P\in\QQ[X_1,\ldots,X_n]$ such that $P(x_1,\ldots,x_n)=0$ is the zero polynomial\footnote{Otherwise, it is called \textit{algebraically dependent.}}. 

\begin{proposition}
    \label{proposition_community}
    Let $A$ be a symmetric $n\times n$ matrix with rational entries and connected support. Let $D$ be a real diagonal $n\times n$ matrix with algebraically independent diagonal entries. Then, $(A,D)$ is a good pair.
\end{proposition}
\begin{proof}
    Let $x_1,\ldots, x_n$ be the entries of $D$, so that $D = \diag(\{x_1,\ldots,x_n\})$. Let $\{\varphi_{k,t}\}_{k=1}^n$ be an analytic eigenbasis as specififed in Lemma \ref{lemma_support}.
    
    By the remark below Proposition \ref{proposition_all-but-finitely}, it suffices to show $(D,A)$ is a good pair. Observe that $H_0=D$ has simple eigenvalues. Thus, by Proposition \ref{proposition_all-but-finitely}, $H_t$ has simple eigenvalues for all but finitely many $t\in\RR$.

    Now, for each $i,k\in\{1,\ldots,n\}$, let $j=d(i,k)$. We show $\varphi_k^{(j)}(i)\ne0$. From Lemma \ref{lemma_support}, we know that $\varphi_k^{(j)}(i)=(C^j)(i,k)$. As $C(i,r)\ne 0$ implies $d(i,r)\le 1$,
    \[ (C^j)(i,k) = \sum_{i_1\sim i}\sum_{i_2\sim i_1}\ldots\sum_{i_{j-1}\sim i_{j-2}}C(i,i_1)C(i_1,i_2)\cdots C(i_{j-2},i_{j-1})C(i_{j-1},k). \]
    For each term in the sum above, $i\sim i_1\sim\ldots\sim i_{j-1}\sim k$ is a path of minimal length from $i$ to $k$. Let $\Gamma(i,k)$ be the set of minimal length paths from $i$ to $k$ in $\supp(A)$. For $\alpha\in\Gamma(i,k)$, we abbreviate
    \[ A_\alpha := A(\alpha(0),\alpha(1))\cdot A(\alpha(1),\alpha(2))\cdots A(\alpha(j-1),\alpha(j)). \]
    Thus,
    \[ \varphi_k^{(j)}(i) = \sum_{\gamma\in\Gamma(i,k)}C_\gamma = \sum_{\gamma\in\Gamma(i,k)}A_\gamma\prod_{r\in\gamma}(x_k-x_r)^{-1}. \]
    Assume for the sake of contradiction that $\varphi_k^{(j)}(i)=0$. Then, clearing denominators, we obtain that
    \[ R(x_1,\ldots,x_n) := \sum_{\gamma\in\Gamma(i,k)}A_\gamma\prod_{r\notin\gamma}(x_k-x_r) = 0. \]
    Observe that $R$ is a polynomial with rational coefficients. Since $x_1,\ldots,x_n$ are algebraically independent, it follows $R$ is identically zero. Pick any $\gamma'\in\Gamma(i,k)$. Then, $R$ evaluated at
    \[ \begin{cases}
        X_r = 1, & \text{if } r\in\gamma' \\
        X_r = 0, & \text{otherwise}
    \end{cases} \]
    equals $A_{\gamma'}\ne0$. Thus, $R$ is not identically $0$, giving the required contradiction.

    Since $\varphi_k^{(j)}(i)$ is nonzero, we conclude that $\varphi_{k,t}(i)$ is not identically $0$. Therefore, by Proposition \ref{proposition_all-but-finitely}, $H_t$ has non-vanishing eigenvectors for all but finitely many $t\in\RR$.
\end{proof}

\subsection{Proof of Theorem \ref{theorem_continuous-prob}}

With the above proposition, we can now prove our main result for continuous probability distributions on the entries of $V$.

\theoremcontinuousprob*

First, recall the following result on algebraically independent sets:

\begin{lemma}
\label{lemma_algebraic-independence}
    Let $B:=\{(x_1,\ldots,x_n)\in\RR^n\mid \{x_1,\ldots,x_n\} \text{ is algebraically dependent}\}$. Then, $B$ has Lebesgue measure $0$.
\end{lemma}
\begin{proof}
    For a nonzero polynomial $P\in\QQ[X_1,\ldots,X_n]$, denote by $Z(P)$ its zero set. According to \cite{zeroset}, $Z(P)$ has measure $0$. Then, $B = \bigcup_{P}Z(P)$
    is a countable union of measure $0$ sets, so it has measure $0$ itself.
\end{proof}

Our result is now immediate:

\begin{proof}[Proof of Theorem \ref{theorem_continuous-prob}]
    Suppose $V$ is bad. Then, Proposition \ref{proposition_community} implies that $V$ has algebraically dependent diagonal entries. However, by Lemma \ref{lemma_algebraic-independence}, the set of such $V$ has Lebesgue measure $0$. Since $\mu$ is absolutely continuous with respect to the Lebesgue measure, the result follows.
\end{proof}

\section{Symmetry Results}
\label{section_symmetry}

In the case of discrete probability distributions, the probability of obtaining a bad potential will, in general, be nonzero. We find bounds on these probabilities by finding conditions that will make a pair $(A,B)$ bad. We begin by proving the following useful extension of Lemma \ref{lemma_kato-rellich}:

\begin{lemma}
    \label{lemma_kato-rellich vectors for symmetries}
    Let $A$ and $B$ be real symmetric $n\times n$ matrices and let $O$ be an orthogonal $n\times n$ matrix that commutes with both $A$ and $B$. Then there exists an orthonormal basis $\{\varphi_{k,t}\}_{k=1}^n$ of simultaneous complex eigenvectors of $H_t(A,B)$ and $O$ that depends analytically on $t$.
\end{lemma}
\begin{proof}
    For this proof we work over the complex vector space $\CC^n$. Since $O$ commutes with $A$ and $B$, it also commutes with $H_t(A,B)$ for all $t$. As $O$ and $H_t$ are both normal matrices and they commute, the eigenspaces of $O$ are $H_t$ invariant. That is, if $\{U_i\}_{i=1}^k$ are the eigenspaces of $O$, then $H_t$ is block diagonal in the decomposition
    \[ \CC^n = \bigoplus_{i=1}^kU_i. \]
    Let $H^i_t := H_t|_{U_i}$ and $d_i:=\dim(U_i)$. Since $H_t$ is real symmetric, $H^i_t: U_i \rightarrow U_i$ is Hermitian. Thus, by Lemma \ref{lemma_kato-rellich}, $H^i_t$ admits an orthonormal basis of eigenvectors $\{\varphi^i_{j,t}\}_{j=1}^{d_i}$ that depends analytically on $t$. These are also eigenvectors of $H_t$. Because they lie in the eigenspace $U_i$, they are also eigenvectors of $O$. As $O$ is normal, the eigenspaces $U_i$ are orthogonal, so the union 
    \[ \bigcup_{i=1}^k\{ \varphi^i_{j,t} \}_{j=1}^{d_i} \]
    is an orthonormal basis of simultaneous eigenvectors of $H_t$ and $O$ that depends analytically on $t$.
\end{proof}

We can now study how such \textit{shared symmetries} of $A$ and $B$ impact our spectral conditions:

\begin{lemma}
\label{lemma_symmetries}
    Let $A$ and $B$ be real symmetric $n\times n$ matrices. Assume there exists an orthogonal matrix $O \ne I$ that commutes with both $A$ and $B$.
    \begin{enumerate}[(1)]
        \item If $O^2 \ne I$, then $H_t(A,B)$ has degenerate spectrum for all $t\in\RR$.
        \item If there is some $j$ such that $Oe_j = e_j$, then $H_t(A,B)$ has an eigenvector $\varphi_t$ analytic in $t$ that vanishes at $j$ for all $t\in\RR$. 
    \end{enumerate}
\end{lemma}

\begin{proof}
    Let $\{\varphi_{k,t}\}_{k=1}^n$ be the set of shared orthonormal eigenvectors of $O$ and $H_t$ given by lemma \ref{lemma_kato-rellich vectors for symmetries}. Assume that $H_t$ has simple eigenvalues for some $t\in\RR$. Because $H_t$ is nondegenerate and real symmetric, all its eigenvectors are real up to a phase. So, without loss of generality, $\varphi_{k,t}(i)$ is real for all $i$ and $k$. Let $\{\lambda_k\}_{k=1}^n$ be the eigenvalues of $O$ corresponding to $\{\varphi_{k,t}\}_{k=1}^n$. Since $\varphi_k \ne 0$, there exists a vertex $j$ such that $\varphi_k(j) \ne 0$. This gives that
    \[ \lambda_k = \frac{(O\varphi_k)(j)}{\varphi_k(j)} \in \RR \]
    Because $O$ is orthogonal, $|\lambda_k| = 1$, so $\lambda_k \in\{-1,1\}$. Then, all the eigenvalues of $O^2$ are $1$, so $O^2=I$. This proves (1).

    Now, assume that there is some $j$ such that $Oe_j=e_j$. Since $O\ne I$, there is some $\varphi_{k,t}$ whose eigenvalue $\lambda$ is not $1$. Since $O$ is orthogonal, we have
    \[ (\lambda-1)\langle\varphi_{k,t},e_j\rangle = \langle O\varphi_{k,t},Oe_j\rangle-\langle\varphi_{k,t},e_j\rangle = 0, \]
    so $\varphi_{k,t}(j)=0$ at all $t$. This proves (2).
\end{proof}

This lemma immediately gives the following useful result:
\begin{corollary}
\label{proposition_sufficient-for-failure}
    Let $A$ and $B$ be real symmetric $n\times n$ matrices with $n$ odd. If there exists a permutation matrix $P \ne I$ that commutes with both $A$ and $B$, then $(A,B)$ is a bad pair.
\end{corollary}
\begin{proof}
     If $P^2 \ne \Id$, we are done by lemma \ref{lemma_symmetries} (1). If $P^2 = \Id$, $P \ne I$ is a permutation matrix of order $2$, so it can be decomposed as a product of disjoint two-cycles. A product of disjoint two-cycles acts nontrivially on an even number of standard basis vectors, and $n$ is odd, so $P$ must have fix some standard basis vector $e_j$. By lemma \ref{lemma_symmetries} (2), $(A,B)$ is a bad pair. 
\end{proof}

\section{Discrete Probability Distributions}
\label{section_discrete}


We now return to the Anderson Model. In particular, we will consider the Bernoulli distribution 
\[ \mu = \begin{cases}
    1 \quad \text{with probability $p$}\\
    -1 \quad \text{with probability $1-p$}
\end{cases} \]
In contrast to an absolutely continuous distribution, the probability of getting a good potential is no longer\footnote{Take $V=I$. Then, $H_t(V)$ is degenerate for all $t\in\RR$.} $1$. We now prove Theorem \ref{theorem_lower-bound-general}, which gives a rough lower bound on the probability of seeing a bad potential:

\theoremlowerboundgeneral*

\begin{proof}
    For each $i = 1, \dots, d$, let $S_i$ be the permutation matrix that maps the vertex $(r_1,\ldots,r_i,\ldots,r_d)$ to the vertex $(r_1,\ldots,r_i+1,\ldots,r_d)$, where the coordinates are taken modulo the appropriate side length. 
    We say that a potential $V$ has $S_i$ symmetry if it commutes with $S_i$. Note that we can write $\Delta$ as 
    \[ \Delta = 2dI - \sum_{i=1}^d (S_i + S_i^{-1}). \]
    Since shifts in different directions commute with each other, $S_i$ commutes with $\Delta$. 
    As $S_i$ has order $L_i>2$, Lemma \ref{lemma_symmetries} (1) guarantees that any potential with $S_i$ symmetry is bad. Hence,
    \[ \PP(V\text{ is bad}) \ge \PP(V\text{ has }S_i\text{ symmetry for some } i). \]
    Let $S(j_1,\ldots,j_k)$ be the event that $V$ has $S_i$ symmetry for $i=j_1,\ldots,j_k$ By the principle of inclusion-exclusion,
    \begin{align}
    \label{equation_pie}
        \PP(V\text{ has }S_i\text{ symmetry for some } i) = \sum_{k=1}^d\left[(-1)^{k-1}
        \cdot\sum_{j_1<\ldots<j_k}\PP(S(j_1,\ldots,j_k))\right]
    \end{align}     
    We now compute $S(j_1,\ldots,j_k)$. Let $\{j_{k+1},\ldots,j_d\}=\{1,\ldots,d\}\setminus\{j_1,\ldots,j_d\}$. Observe that a potential that has $S_i$ symmetry for $i=j_1,\ldots,j_d$ is entirely determined by its values on the subgraph $\ZZ/L_{j_{k+1}}\ZZ\times\cdots\times\ZZ/L_{j_d}\ZZ$ defined by the tuples that are $0$ on the $j_1,\ldots,j_k$ dimensions. The probability of getting such a potential with exactly $m$ positive $1$ entries in the subgraph is
    \[ \binom{L_{j_{k+1}}\cdots L_{j_d}}{m}\cdot p^{mL_{j_1}\cdots L_{j_k}}\cdot(1-p)^{(L_{j_{k+1}}\cdots L_{j_d}-m)L_{j_1}\cdots L_{j_k}}. \]
    Adding over $m$, we get
    \begin{align*}
        \PP(S(j_1,\ldots,j_k)) = \left(p^{L_{j_1}\cdots L_{j_k}} + (1-p)^{L_{j_1}\cdots L_{j_k}}\right)^{L_{j_{k+1}}\cdots L_d}.
    \end{align*}
    Plugging this into equation (\ref{equation_pie}) we obtain the claimed lower bound.
\end{proof}

It is worth noting that this lower bound can be improved by considering other shared symmetries, such as reflections.
In general, we do not know whether the converse of lemma \ref{lemma_symmetries} is true. So, we do not get an upper bound for $\PP(V \, \text{is bad})$. However, in the one-dimensional case when $L$ is prime, we see that every bad potential has a reflection symmetry. 



\begin{proposition}
\label{Proposition-L-prime}
    Suppose $d=1$ and $L$ is prime. Then $V=\mathrm{diag}(v(1),\ldots,v(n))$ is bad if and only if $v$ has reflection symmetry $v(j+i)=v(j-i)$ about some vertex $j$.
    In particular, 
    \begin{enumerate}[(1)]
        \item If $H_t$ has degenerate spectrum for all $t$, then $V=\pm I$. 
        \item If $H_t$ has some vanishing eigenvector for all $t$, then $v$ has some reflection symmetry.
    \end{enumerate}
\end{proposition}
\begin{proof}
    Throughout this proof, we take indices modulo $L$ for simplicity. First, notice that if $v$ has a reflection symmetry about a vertex $j$, then $V$ shares a permutation symmetry with $\Delta$. Because $L$ is odd, corollary \ref{proposition_sufficient-for-failure} guarantees that any such $V$ is bad. It remains to prove the converse.

    Lemma \ref{lemma_kato-rellich} gives an orthonormal basis of eigenvectors $\{\varphi_{k,t}\}_{k=0}^{L-1}$ and corresponding eigenvalues $\{\lambda_{k,t}\}_{k=0}^{L-1}$ that depend analytically on $t$. Let $\omega = \exp(2\pi i/L)$. The \textit{Discrete Fourier Transform} is the unitary matrix defined as
    \[ \SF(r,s) := \frac1{\sqrt L}\omega^{rs}. \]
    Observe that $\SF$ diagonalizes $\Delta$. Namely, we have that
    \[ \tilde\Delta(r,s) := \left(\SF\Delta\SF^{-1}\right)(r,s) = \delta_{r,s}(2-\omega^r-\omega^{-r}). \]
    Note that $0$ is a simple eigenvalue of $\Delta$. Recall that $L$ is an odd prime, so every other eigenspace has dimension $2$, corresponding to $\tilde\Delta(r,r) = \tilde\Delta(-r,-r)$.
    Similarly, we compute
    \[ \tilde V(r,s) := \left(\SF V\SF^{-1}\right)(r,s) = \frac1L\sum_{j=0}^{L-1}\omega^{j(r-s)}v(j). \]
    Define $$P(z) := \frac1L\sum_{j=0}^{L-1}v(j)z^j,$$
    so $\tilde V(r,s) = P(\omega^{r-s})$. Consider the one-parameter perturbation $\tilde H_t = \tilde\Delta + t\tilde V$ and define $\tilde\varphi_{k,t} = \SF\varphi_{k,t}$. Since $\SF$ is unitary, we know that $\{\tilde\varphi_{k,t}\}_{k=0}^{L-1}$ is an orthonormal basis and
    \begin{align}
        \label{perturbation-eq}
        \tilde H_t \tilde\varphi_{k,t} = \lambda_{k,t}\tilde\varphi_{k,t}.
    \end{align}
    Let $\tilde\varphi_k^{(0)} = \tilde\varphi_{k,0}$ and $\lambda_k^{(0)}=\lambda_{k,0}=2-\omega^k-\omega^{-k}$.
    Denote by $\tilde\varphi_k^{(1)},\lambda_k^{(1)}$ the first derivatives at $t=0$ of $\tilde\varphi_{k,t}$ and $\lambda_{k,t}$ respectively. Taking the first derivative of equation (\ref{perturbation-eq}) at $t=0$, we obtain
    \begin{align}
        \label{first-order-eq}
        \tilde V\tilde\varphi_k^{(0)} + \tilde\Delta\tilde\varphi_k^{(1)} = \lambda_k^{(1)}\tilde\varphi_k^{(0)} + \lambda_k^{(0)}\tilde\varphi_k^{(1)}.
    \end{align}
    Suppose that $H_t$ has degenerate spectrum for all $t$. It follows from Proposition \ref{proposition_all-but-finitely} that there is some $k\ne0$ such that $\lambda_{k,t}=\lambda_{-k,t}$ for all $t$. In particular, $\lambda_k^{(1)} = \lambda_{-k}^{(1)}$. Evaluating equation (\ref{first-order-eq}) at the $k-$th and $(-k)-$entries, we obtain
    \begin{align}
    \label{eigval-problem}
        \begin{pmatrix}
            P(1)& P(\omega^{2k}) \\
            P(\omega^{-2k}) & P(1)
        \end{pmatrix} \begin{pmatrix}
            \tilde\varphi_k^{(0)}(k) \\
            \tilde\varphi_k^{(0)}(-k)
        \end{pmatrix} = \lambda_k^{(1)}\begin{pmatrix}
            \tilde\varphi_k^{(0)}(k) \\
            \tilde\varphi_k^{(0)}(-k)
        \end{pmatrix}.
    \end{align}
    Note that $\lambda_{-k}^{(1)}$ solves the same eigenvalue problem (\ref{eigval-problem}) as $\lambda_k^{(1)}$. Therefore, $\lambda_k^{(1)}=\lambda_{-k}^{(1)}$ implies that $P(\omega^{2k})=0$. Since the minimal polynomial of $\omega^{2k}$ is $\Phi:=\sum_{j=0}^{L-1}z^j$, it follows that $P=\pm L\Phi$. Thus, $V=\pm I$. This proves (1).

    Now, suppose that $H_t$ has some vanishing eigenvector for all $t$, so $\varphi_{k,t}(j)=0$ for all $t$. Since (2) clearly holds for $V= \pm I$, we may assume that $P(\omega^{2k})\ne0$, and we can solve the eigenvalue problem (\ref{eigval-problem}) to obtain
    \begin{align}
    \label{equation_eigvec}
        \begin{pmatrix}
            \tilde\varphi_{k}^{(0)}(k) \\
            \tilde\varphi_{k}^{(0)}(-k)
        \end{pmatrix} = \frac1{\sqrt2}\begin{pmatrix}
            \frac{P(\omega^{2k})}{|P(\omega^{2k})|} \\
            \pm1
        \end{pmatrix}.
    \end{align}
    Recall that $\varphi_k^{(0)} = \SF^{-1}\tilde\varphi_k^{(0)}$. Combining this with $\varphi_k^{(0)}(j)=0$ and equation (\ref{equation_eigvec}) gives us the condition
    \[ \omega^{-2jk}P(\omega^{2k}) = \pm|P(\omega^{2k})| \in\RR. \]
    Expanding the imaginary part, we have 
    \begin{align*}
        0 &= \sum_{r=0}^{L-1}v(r)\left(\omega^{2k(r-j)}-\omega^{2k(j-r)}\right) \\
        &= \sum_{r=0}^{L-1}\left(v(j+(2k)^{-1}r)-v(j-(2k)^{-1}r)\right)\omega^r,
    \end{align*}
    where $(2k)^{-1}$ is the inverse of $2k$ modulo $L$. Define 
    \[ f(z):=\sum_{r=0}^{L-1}(v(j+(2k)^{-1}r)-v(j-(2k)^{-1}r))z^r. \] Since the minimal polynomial of $\omega$ is $\Phi$ and the constant coefficient of $f$ is $0$, we must have that $f$ is the zero polynomial. Therefore, $v(j+(2k)^{-1}r)-v(j-(2k)^{-1}r) = 0$ for all $r$. Letting $r = 2ki$, this proves (2).

\end{proof}

We use this proposition to prove Theorem \ref{theorem_exact-probability}:

\theoremexactprobability*

\begin{proof}
    By Proposition \ref{Proposition-L-prime}, $V$ is bad if and only if it has a reflection symmetry about a vertex $j$. Then, $V$ is determined by its value at $j$ and at the $\frac{L-1}{2}$ vertices to one side of vertex $j$. The probability of getting such a potential with $+1$ at exactly $2m$ vertices not including the vertex $j$ is 
    \[ \binom{\frac{L-1}{2}}{m}p^{2m}(1-p)^{L-1-2m} \] 
    Summing these over $m$,
    \begin{align*}
        \PP(V \text{ has reflection about $j$}) = (p^2+(1-p)^2)^\frac{L-1}{2}.
    \end{align*}

    Assume that $V$ has a reflection symmetry about two vertices $j$ and $k$. Then,
    \[ V(r,r) = V(2j-r,2j-r) = V(2k-2j+r,2k-2j+r) \]
    Since $2k-2j$ is nonzero modulo $L$, the above equality iterated $(2k-2j)^{-1}(s-r)$ times gives that $V(r,r)=V(s,s)$ for any $r$ and $s$, so $V=\pm I$. The probability this happens is $p^L+(1-p)^L$. Summing over the vertices, 
    \begin{align*}
        \PP(V \text{ is bad}) &= \sum_{j=1}^L\left((p^2+(1-p)^2)^\frac{L-1}{2}-(p^L+(1-p)^L)\right) + p^L+(1-p)^L\\
        &= L(p^2+(1-p)^2)^\frac{L-1}{2} - (L-1)(p^L+(1-p)^L)
    \end{align*}

    If $p \not= 0,1$, then $|p|<1, |1-p|<1$, and $|p^2+(1-p)^2|<1$. So, taking the limit as $L \rightarrow \infty$, we have 
    \begin{align*}
        \lim_{L \rightarrow \infty}\PP(V \text{ is bad}) &= 
        0.
    \end{align*}
\end{proof}

\pagebreak

\section{Acknowledgments}

This paper is the culmination of a project started for the 2025 edition of the Summer Program in Undergraduate Research (SPUR) at the Massachusetts Institute of Technology. The authors thank Lior Alon for proposing the project and his valuable mentorship throughout the program. We also thank Roman Bezrukavnikov and Jonathan Bloom for organizing SPUR and their feedback during the summer. Finally, we thank the Department of Mathematics for providing funding for the project.

\bibliographystyle{amsalpha}
\bibliography{bibspur}

\providecommand{\bysame}{\leavevmode\hbox to3em{\hrulefill}\thinspace}
\providecommand{\MR}{\relax\ifhmode\unskip\space\fi MR }
\providecommand{\MRhref}[2]{%
  \href{http://www.ams.org/mathscinet-getitem?mr=#1}{#2}
}
\providecommand{\href}[2]{#2}
\begin{thebibliography}{LMNN20}

\bibitem[ACTA73]{Abou-Chacra_1973}
R~Abou-Chacra, D~J Thouless, and P~W Anderson, \emph{A selfconsistent theory of localization}, Journal of Physics C: Solid State Physics \textbf{6} (1973), no.~10, 1734.

\bibitem[AG23]{AlonGoresky}
Lior Alon and Mark Goresky, \emph{Morse theory for discrete magnetic operators and nodal count distribution for graphs}, Journal of Spectral Theory (2023).

\bibitem[And58]{anderson1958absence}
Philip~W Anderson, \emph{Absence of diffusion in certain random lattices}, Physical review \textbf{109} (1958), no.~5, 1492.

\bibitem[AU24]{alon2024average}
Lior Alon and John Urschel, \emph{Average nodal count and the nodal count condition for graphs}, arXiv preprint arXiv:2404.03151 (2024).

\bibitem[AW09]{aizenmanWarzel2009localization}
Michael Aizenman and Simone Warzel, \emph{Localization bounds for multiparticle systems}, Communications in Mathematical Physics \textbf{290} (2009), no.~3, 903--934.

\bibitem[Car05]{zeroset}
Richard Caron, \emph{The zero set of a polynomial}, 2005.

\bibitem[DS19]{Ding_Smart_2019}
Jian Ding and Charles~K. Smart, \emph{Localization near the edge for the anderson bernoulli model on the two dimensional lattice - inventiones mathematicae}, Jul 2019.

\bibitem[EK16]{ELGART20163465}
Alexander Elgart and Abel Klein, \emph{An eigensystem approach to anderson localization}, Journal of Functional Analysis \textbf{271} (2016), no.~12, 3465--3512.

\bibitem[Kat13]{kato2013perturbation}
Tosio Kato, \emph{Perturbation theory for linear operators}, vol. 132, Springer Science \& Business Media, 2013.

\bibitem[Kle07]{klein2007multiscaleanalysislocalizationrandom}
Abel Klein, \emph{Multiscale analysis and localization of random operators}, 2007.

\bibitem[LMNN20]{logunov2020landisconjectureexponentialdecay}
A.~Logunov, E.~Malinnikova, N.~Nadirashvili, and F.~Nazarov, \emph{The landis conjecture on exponential decay}, 2020.

\bibitem[Rel69]{rellich1969perturbation}
Franz Rellich, \emph{Perturbation theory of eigenvalue problems}, CRC Press, 1969.

\end{thebibliography}

\end{document}